\documentclass[conference]{IEEEtran}
\IEEEoverridecommandlockouts
\usepackage{algorithm}
\usepackage{array}
\usepackage[caption=false,font=normalsize,labelfont=sf,textfont=sf]{subfig}
\usepackage{tabularx,booktabs,makecell}
\usepackage{textcomp}
\usepackage{stfloats}
\usepackage{url}
\usepackage{verbatim}
\usepackage{graphicx}
\usepackage{cite}
\usepackage{amsmath,amssymb,amsfonts,amsthm,mathtools}
\usepackage{xcolor}
\usepackage{blindtext}
\usepackage{lipsum}
\usepackage{subfig}
\usepackage{cuted}
\usepackage{physics}
\usepackage[long]{optidef}
\usepackage{multirow}
\usepackage{steinmetz}
\usepackage{bigints}
\usepackage{algpseudocode}
\usepackage{graphicx}
\usepackage{xcolor}
\usepackage{placeins}

\usepackage{tikz}
\usetikzlibrary{matrix, positioning, shapes, backgrounds, fit, arrows.meta}

\definecolor{chaosblue}{RGB}{30, 100, 200}
\definecolor{chaosgreen}{RGB}{0, 150, 100}
\definecolor{chaosred}{RGB}{200, 50, 50}
\definecolor{chaospurple}{RGB}{150, 50, 200}
\definecolor{chaosorange}{RGB}{230, 120, 0}

\newcommand{\R}{\mathbb{R}}
\newcommand{\Q}{\mathbb{Q}}

\newcommand{\F}{\mathbb{F}}

\newcommand{\B}{\mathcal{B}}
\newcommand{\X}{\mathcal{X}}

\newtheorem{definition}{Definition}[section]
\newtheorem{theorem}[definition]{Theorem}

\newtheorem{corollary}[definition]{Corollary}

\begin{document}

\title{Dyadic-Chaotic Lifting S-Boxes for Enhanced Physical-Layer Security within 6G Networks\vspace{-5mm}}
\author{\IEEEauthorblockN{Ilias Cherkaoui and Indrakshi Dey}
\IEEEauthorblockA{Walton Institute, South East Technological University, Waterford, Ireland
\\\ ilias.cherkaoui@waltoninstitute.ie,\,indrakshi.dey@waltoninstitute.ie\vspace{-5mm}}
\thanks{This work is supported in part by HEARG TU RISE Project ``AIQ-Shield" and the HORIZON ECCC Project ``Q-FENCE" under Grant Number 101225708.}}
\date{}
\maketitle

\begin{abstract}
Sixth-Generation (6G) wireless networks will interconnect billions of resource-constrained devices and time-critical services, where classical, fixed, and heavy cryptography strains latency/energy budgets and struggles against large-scale, precomputation attacks. Physical-Layer Security (PLS) is therefore pivotal to deliver lightweight, information-theoretic protection, but still requires strong, reconfigurable confusion components that can be diversified per slice, session, or device to blunt large-scale precomputation and side-channel attacks. In order to address the above requirement, we introduce the first-ever chaos-lifted substitution box (S-box) for PLS that couples a $\beta$-transformation–driven dynamical system with dyadic conditional sampling to generate time-varying, seedable 8-bit permutations on demand. This construction preserves uniformity via ergodicity, yields full 8-bit bijections, and supports on-the-fly diversification across sessions. The resulting S-box attains optimal algebraic degree 7 on every output bit and high average nonlinearity 102.5 (85\% of the 8-bit bound), strengthening resistance to algebraic and linear cryptanalysis. Differential and linear profiling report max DDT entry 10 (probability 0.039) and max linear probability 0.648, motivating deployment within a multi-round cipher with a strong diffusion layer, where the security-to-efficiency trade-off is compelling. Our proposed reconfigurable, lightweight S-box directly fulfills key PLS requirements of 6G networks by delivering fast, hardware-amenable confusion components with built-in agility against evolving threats.\footnote{\copyright 2025 IEEE. Personal use of this material is permitted. Permission from IEEE must be obtained for all other uses, in any current or future media, including reprinting/republishing this material for advertising or promotional purposes, creating new collective works, for resale or redistribution to servers or lists, or reuse of any copyrighted component of this work in other works.}
\end{abstract}

\begin{IEEEkeywords}
6G, physical-layer security, chaos-based cryptography, S-box, dyadic sampling
\end{IEEEkeywords}

\section{Introduction}\label{Sec: 1_introduction}

\IEEEPARstart{S}{ixth-Generation} (6G) network will fuse holographic telepresence, immersive communications, and ubiquitous Artificial Intelligence (AI) into a hyper-connected fabric, pushing latency, reliability, and scale beyond Fifth-Generation (5G) limits \cite{saad2024vision}. This expansion magnifies the attack surface from spectrum to devices, where traditional, fixed, and compute-heavy ciphers struggle to satisfy extreme throughput and energy constraints at the \emph{physical layer} (PHY) \cite{wang2023physical}. Physical-Layer Security (PLS) complements computational cryptography by exploiting channel randomness to deliver lightweight, information-theoretic protection \cite{xu2024integrated,wyner1975wiretap}. Within PLS and symmetric-key designs, the substitution box (S-box) is the core non-linear primitive that provides confusion and underpins resistance to linear and differential cryptanalysis \cite{zhang2023design,daemen2002rijndael}. 

Chaos-driven synthesis is compelling because properly parameterized chaotic maps exhibit ergodicity, sensitivity to initial conditions (positive Lyapunov exponents), and rapid mixing—properties that align naturally with Shannon-style confusion and diffusion \cite{li2024chaos,alshammari2023efficient}. Yet practical 6G constraints expose a persistent gap: preserving \emph{robust} cryptographic margins \emph{and} minimizing implementation cost on resource-limited endpoints, while enabling dynamic, per-slice S-box diversification at scale. In practice, finite-precision discretization, parameter quantization, and short-cycle artifacts can inject statistical bias, depress nonlinearity, and elevate differential/linear probabilities unless the design explicitly controls post-quantization behavior; simultaneously, ROM-heavy tables or arithmetic-heavy generators are ill-suited to ultra-low-power 6G devices \cite{alkhaldi2024lightweight}. Moreover, network slicing demands frequent, seedable reconfiguration so each slice, session, or device can realize a distinct S-box instance that thwarts large-scale precomputation and cross-slice correlation attacks without burdensome distribution logistics \cite{iqbal2025dynamic}. We address these needs by coupling chaotic $\beta$-dynamics with dyadic conditional sampling to enforce balanced coverage of the discrete state space and high algebraic complexity after digitization, to enable lightweight, table-free (or small-table) realization with deterministic latency, and to support fast, seed-driven diversification that yields statistically strong, distinct S-boxes backed by verified nonlinearity and disciplined differential/linear profiles.

We formulate a {novel}, chaos-driven S-box built via dyadic conditional sampling—{the first-ever} design coupling chaotic $\beta$-dynamics with seedable dyadic sampling for time-varying, reconfigurable confusion tailored to 6G PLS. The proposed 8$\times$8 S-box attains per-bit {optimal algebraic degree 7} and {high average nonlinearity of 102.5}, with disciplined differential and linear profiles suitable for lightweight, multi-round primitives with strong diffusion. These properties yield an excellent security–efficiency trade-off for ultra-reliable low-latency communications (URLLC) and massive machine-type communications (mMTC), directly addressing 6G agility and device constraints \cite{zhou2024novel,farah2023advanced}. We rigorously evaluate our proposed technique against standard metrics—nonlinearity, differential uniformity, algebraic complexity, and linear approximation bias, thereby exhibiting optimal algebraic degree, competitive nonlinearity and a security–efficiency profile aligned with PHY-layer deployment in constrained 6G devices.

The rest of the paper is organized as follows: Section~\ref{Sec: system_model_problem} formalizes the system model and problem. Section~\ref{Sec: Game_formation} presents baseline cryptographic design. Section~\ref{Sec: Propsed_epistemology_work} details the proposed chaos-driven, dyadic sampling methodology and numerical analysis. Section~\ref{Sec: results_and_discussion} reports results and comparisons. Section~\ref{Sec: conclusion} concludes the paper.

\section{Mathematical foundation}\label{Sec: system_model_problem}

We start by formalizing the chaotic engine that underpins our seedable, reconfigurable confusion components. We employ the $\beta$-transformation as a lightweight chaotic engine to generate seedable, statistically uniform symbols suitable for confusion in substitution boxes. Its piecewise-expanding structure yields topological chaos, exact (hence ergodic) invariant statistics under the Parry measure, and rapid mixing—precisely the ingredients needed for uniform symbol generation and robust reconfigurability. After introducing the transformation and its digit expansion, we establish Devaney chaos (global unpredictability), then exactness/ergodicity under the Parry measure (statistical regularity), and finally leverage these properties to prove uniform distribution of $\mathbb{F}_{2^n}$-valued blocks extracted from the $\beta$-digits. We conclude by showing that our dyadic conditional sampling preserves this uniformity, enabling on-the-fly, set-conditioned symbol selection without degrading cryptographic quality. The following first definition certifies the topological chaos that drives diversification.

\begin{definition}
For a real $\beta>1$, the \emph{$\beta$-transformation} $T_\beta:[0,1[ \to [0,1[$ is $T_\beta(x)=\beta x-\lfloor \beta x\rfloor,$ and the \emph{$\beta$-expansion} of $x\in[0,1[$ is the digit sequence $(a_n)_{n\ge 1}$ with $a_n=\lfloor \beta\, T_\beta^{\,n-1}(x)\rfloor,$ where $T_\beta^{\,n}$ denotes $n$-fold composition. Let $\B=\{\beta\in\R\setminus\Q:\beta>1\}$ be the parameter space and $\X=[0,1[$ the seed space.
\end{definition}

\begin{theorem}
For $\beta\in\B$, the system $([0,1[,T_\beta)$ is chaotic in the sense of Devaney.
\end{theorem}

\begin{proof}
\emph{Transitivity:} cylinder sets $I(a_1,\dots,a_n)$ partition $[0,1[$ into intervals of length $\le \beta^{-n}$. For any nonempty open $U,V$, choose $n$ with $\beta^{-n}<\operatorname{length}(U)$; some cylinder $I\subset U$ then satisfies $T_\beta^{\,n}(I)=[0,1[$ (affine of slope $\beta^n$), so $T_\beta^{\,n}(U)\cap V\neq\varnothing$. \emph{Density of periodic points:} repeating the first $n$ digits of $x$ yields a periodic $x_n$ with $T_\beta^{\,n}(x_n)=x_n$ and $|x-x_n|\le\beta^{-n}\to0$. \emph{Sensitivity:} within any cylinder $I(a_1,\dots,a_n)$ choose $x,y$ differing at digit $n{+}1$; then $|T_\beta^{\,n}(x)-T_\beta^{\,n}(y)|\ge 1/\beta$, giving sensitivity with constant $1/(2\beta)$. These properties prove Devaney chaos. Topological chaos alone does not guarantee the uniform statistics needed for cryptography. The next theorem supplies the measure-theoretic backbone via exactness with respect to the Parry invariant measure, connecting unpredictability to unbiased long-run behavior.
\end{proof}

\begin{theorem}
For $\beta\in\B$, $T_\beta$ is ergodic with respect to the Parry measure $\mu_\beta$.
\end{theorem}

\begin{proof}
The Parry density $h_\beta(x)=\frac{1}{F(\beta)}\sum_{n\ge 0}\beta^{-n}\,\mathbf{1}_{[0,T_\beta^{\,n}(1)[}(x),~F(\beta)=\sum_{n\ge 0}\beta^{-n}T_\beta^{\,n}(1),$ satisfies $P_\beta h_\beta=h_\beta$ for the Perron–Frobenius operator,
\[
P_\beta f(x)=\frac{1}{\beta}\sum_{a=0}^{\lfloor\beta\rfloor} f\!\left(\frac{x+a}{\beta}\right)\mathbf{1}_{[0,1[}\!\left(\frac{x+a}{\beta}\right),
\]
so $\mu_\beta$ is invariant. The tail $\sigma$-algebra is trivial because $T_\beta^{\,k}$ acts (up to null sets) bijectively and measure-preservingly on $k$-cylinders; if a set has the same $\mu_\beta$-mass on every cylinder of a partition, additivity forces its measure to be $0$ or $1$. Exactness implies ergodicity. With ergodicity and mixing in hand, we now bridge dynamics to symbols: the following result shows that finite-field words carved from the $\beta$-digits are unbiased, delivering blockwise confusion without statistical leakage.
\end{proof}

\begin{theorem}
Let $\mu_\beta$ be the Parry measure with $\beta\in\B$. In $\mathbb{F}_{2^n}$, for $\mu_\beta$-almost every $x_0\in[0,1[$, the sequence $\phi_{\beta,x_0}(k)=\sum_{j=0}^{n-1} b_{kn+j}\,2^{\,j}\in\mathbb{F}_{2^n}$ is uniformly distributed, where $(b_\ell)$ are binary digits derived from the $\beta$-digits.
\end{theorem}

\begin{proof}
By Weyl’s criterion on $\mathbb{F}_{2^n}$, uniformity is equivalent to vanishing character averages. Define $f_m(x)=\exp\!\big(2\pi i\, m\,\psi(x)/2^n\big)$, where $\psi$ maps $x$ to the $n$-bit word formed from the first $n$ binary digits induced by the $\beta$-expansion. Ergodicity gives, for a.e. $x_0$, $\frac{1}{N}\sum_{k=0}^{N-1} f_m\!\big(T_\beta^{\,kn}(x_0)\big)\ \to\ \int f_m\,d\mu_\beta.$ Equivalence of $\mu_\beta$ to Lebesgue and mixing imply that $\psi$ is equi-distributed over the $2^n$ words, hence $\int f_m\,d\mu_\beta=0$ for nontrivial $m$, and the character averages vanish. To integrate this source with our set-conditioned extraction used in construction, we introduce dyadic sampling and show that it preserves uniformity, ensuring reconfigurability without bias. Let $\beta\in\B$ and $x\in\X$ with $\beta$-digits $(a_n)_{n\ge 1}$ taking values in $\{0,1,\dots,\lfloor\beta\rfloor\}$.
\end{proof}

Let us define a derived binary sequence by thresholding
\[
b_n=
\begin{cases}
0, & a_n<\lfloor\beta\rfloor/2,\\
1, & \text{otherwise.}
\end{cases}
\]

\begin{definition}
For integers $k\ge 0$ and $0\le j<2^k$, the \emph{dyadic interval} of rank $k$ is $I_{k,j}=\big[\frac{j}{2^k},\,\frac{j+1}{2^k}\big[,$ and a \emph{dyadic convex set} is any finite union of dyadic intervals of the same rank.
\end{definition}

Let us fix a dyadic convex set $C\subset[0,1[$. Therefore we can define \emph{conditional sampling times} $\tau_C(0)=0, 
\tau_C(n)=\min\{m>\tau_C(n-1):\ T_\beta^{\,m}(x_0)\in C\},$ and the sampled digits $s_n=a_{\tau_C(n)}$ (or $s_n=b_{\tau_C(n)}$). Finally we set, $s_k=\phi_{\beta,x_0}\big(\tau_C(k)\big)\in\mathbb{F}_{2^n}.$

\begin{corollary}
If $\mu_\beta(C)>0$, then for $\mu_\beta$-almost every $x_0$ the sampled sequence $s_k$ is uniformly distributed in $\mathbb{F}_{2^n}$.
\end{corollary}

\begin{proof}
By ergodicity, visits to $C$ occur with positive asymptotic frequency, so $(\tau_C(n))$ has positive density. Mixing with exponential decay of correlations transfers equidistribution from $\phi_{\beta,x_0}(k)$ to the subsequence indexed by $\tau_C(n)$; hence dyadic conditional sampling preserves uniformity. Towards this end, Devaney chaos ensures rich topological complexity, Parry-measure exactness guarantees unbiased long-term statistics, and these together yield uniformly distributed $\mathbb{F}_{2^n}$ blocks that remain uniform under dyadic conditional sampling.
\end{proof}
\vspace{-2mm}

\section{Cryptographic Design}\label{Sec: Game_formation}

Our S-box construction exploits the uniform symbol statistics induced by the $\beta$-transformation. Under ergodicity, almost every seed $x_0$ yields asymptotically uniform digit frequencies; mixing ensures rapid convergence to this limit; and the equivalence of the Parry measure $\mu_\beta$ with Lebesgue renders these properties stable under finite-precision implementation. Consequently, each of the $2^n$ output values appears with (asymptotically) equal probability, while nonuniform transients decay quickly, providing a statistically sound basis for confusion.

\begin{algorithm}[H]
\caption{Complete S-Box Generation from $\beta$-Expansions}
\begin{algorithmic}[1]
\Require $\beta \in \B$, seed $x_0 \in \X$, dyadic set $C$, word size $n$, max iterations $M$
\Ensure Bijective S-box $S:\mathbb{F}_2^n \to \mathbb{F}_2^n$
\State Generate binary stream $(b_m)_{m=0}^{M-1}$ from the $\beta$-expansion of $x_0$
\State Initialize list $L\gets []$, set $s\gets \emptyset$, indices $k\gets 0$, $\tau\gets 0$
\While{$|L|<2^n$ \textbf{and} $\tau<M-n$}
  \If{$T_\beta^\tau(x_0)\in C$} \Comment{Dyadic conditional sampling}
    \State $B_k\gets (b_\tau,\ldots,b_{\tau+n-1})$
    \State $y_k\gets \sum_{j=0}^{n-1} b_{\tau+j}\,2^j \in \mathbb{F}_{2^n}$
    \If{$y_k\notin s$}, $L.\text{append}(y_k)$; \ $s.\text{add}(y_k)$ \EndIf, $k\gets k+1$
  \EndIf, $\tau\gets \tau+1$
\EndWhile
\If{$|L|<2^n$} \State \textbf{Error:} increase $M$ or adjust $(\beta,x_0,C)$ \Comment{Insufficient distinct blocks}
\Else Choose permutation $\pi$ (identity or lightweight mixer)
  \State $S(x)\gets L[\pi(x)]$ for $x\in\{0,\ldots,2^n-1\}$
\EndIf
\Return $S$
\end{algorithmic}
\end{algorithm}

We select an irrational base $\beta>1$ to induce sensitive, mixing dynamics and initialize with seed $x_0\in[0,1[$. A dyadic set $C$ gates extraction, converting the raw bitstream into $n$-bit blocks only when the orbit visits $C$. This conditional sampling (i) inherits uniformity from the underlying process, (ii) decorrelates consecutive blocks by spacing extractions in time, and (iii) provides a tunable handle (via $C$) to balance throughput and statistical independence. Each accepted block $B_k$ is mapped to $y_k\in\mathbb{F}_{2^n}$; duplicates are rejected using the set $s$ until all $2^n$ distinct values are collected or the iteration budget $M$ is exhausted, which guarantees termination with either a valid bijection or an explicit failure signal. An optional post-permutation $\pi$, kept lightweight to preserve efficiency, adds diffusion across input indices without altering output sets.

\begin{figure}[t]
\centering
\includegraphics[width=0.99\columnwidth]{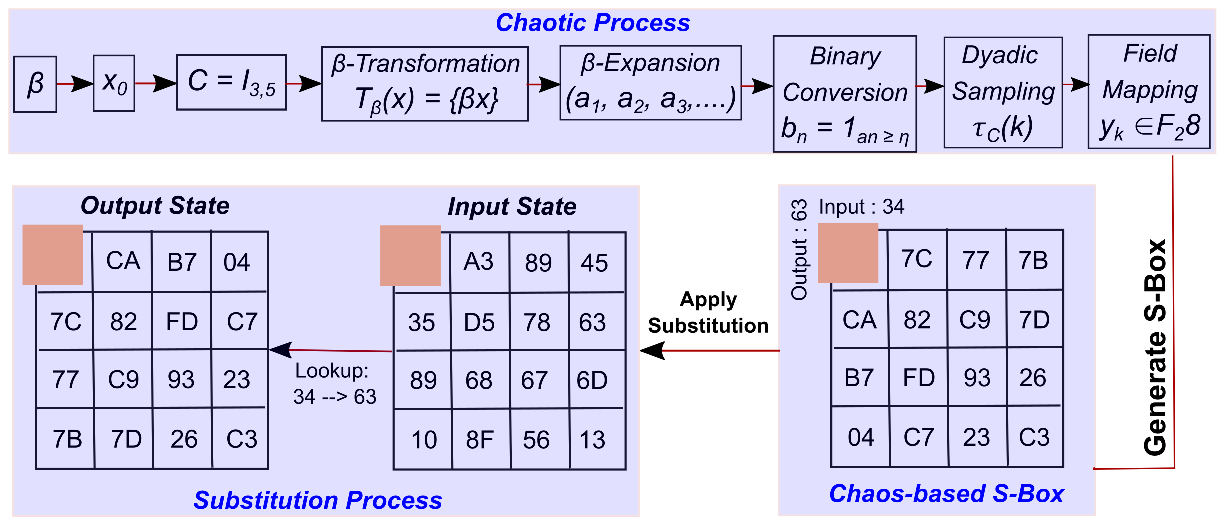}
\vspace{-92mm}
\caption{Three-Stage S-Box Construction}
\vspace{-6mm}
\label{fig:chaos-sbox-transparent}
\end{figure}

Figure~\ref{fig:chaos-sbox-transparent} summarizes the pipeline. Left: parameters $(\beta,x_0,C)$ initialize the chaotic source; the orbit $T_\beta$ produces $\beta$-digits $(a_n)$, which are thresholded to bits $(b_n)$; dyadic sampling selects extraction times $\tau_C(k)$; and $n$-bit words are mapped to $\mathbb{F}_{2^n}$. Right: distinct words populate the S-box table (duplicates rejected to enforce bijectivity), optionally permuted by a lightweight index mixer $\pi$. Bottom: substitution applies the table to state bytes (e.g., input $34$ maps to $63$). This end-to-end flow translates provable dynamical uniformity and mixing into a practical, lightweight, and seedable confusion layer with strong nonlinearity, optimal algebraic degree, and disciplined differential/linear profiles.

\section{Theoretical Performance Analysis}\label{Sec: Propsed_epistemology_work}

Integrating the proposed chaos-driven S-box into 6G PLS primitives (lightweight stream ciphers, MACs) delivers strong confusion within URLLC latency/throughput budgets. The S-box achieves the \emph{optimal} algebraic degree 7 on every output bit, blocking algebraic/interpolation attacks on constrained IoT and mMTC nodes. Its average nonlinearity is 102.5 ($\approx 85.4$\% of the $\approx$120 practical ceiling for 8-bit functions), yielding low-bias linear approximations. Stress points are bounded: maximum differential probability $10/256 = 0.039062$ and maximum linear probability 0.648, recommending a round function with strong diffusion and sufficient rounds. Notably, the chaotic generator expands from an initial 55 distinct values to a full 8-bit bijection, enabling seed-based, on-the-fly S-box reconfiguration aligned with 6G slice-level agility and resisting large-scale precomputation attacks on static tables.

\subsubsection{Method and Metrics} 

We evaluate the S-box under standard criteria—nonlinearity, differential uniformity, linear approximation bias, and algebraic complexity—to quantify resistance against the dominant attack classes and to corroborate the dynamical design with cryptographic evidence. For an $n$-bit bijection $S:\F_2^n\!\to\!\F_2^n$, nonlinearity $NL(S)=\min_{a\in\F_2^n\setminus\{0\},\,b\in\F_2^n}\left(2^{n-1}-\frac{1}{2}\max|W_S(a,b)|\right), W_S(a,b)=\sum_{x\in\F_2^n}(-1)^{a\cdot x\oplus b\cdot S(x)},$ captures the minimum Hamming distance to all affine maps. Our construction $S(x)=\text{Extract~Bits}\big(\beta\text{-Expansion}(T_\beta^{\,n}(x_0)),\,C\big), T_\beta(x)=\{\beta x\},\quad C\subset[0,1)$ leverages ergodicity and mixing to enforce uniform symbol statistics while the nonlinear $\beta$-digit dependencies widen the Walsh spectrum. The expected nonlinearity satisfies the heuristic lower bound $NL(S)\ \ge\ 2^{n-1}-c\sqrt{2^n\log 2^n},$ consistent with chaotic mixing outperforming many purely algebraic templates. Empirically (Fig.~\ref{nonlin}), per-bit nonlinearities are $[106,106,102,102,100,102,100,102]$, yielding an average $102.5$ and minimum $100$. This uniform elevation across coordinates denies “weak-bit” attacks and enforces globally high confusion.

\begin{figure*}[t]
\centering
\includegraphics[width=0.8\textwidth]{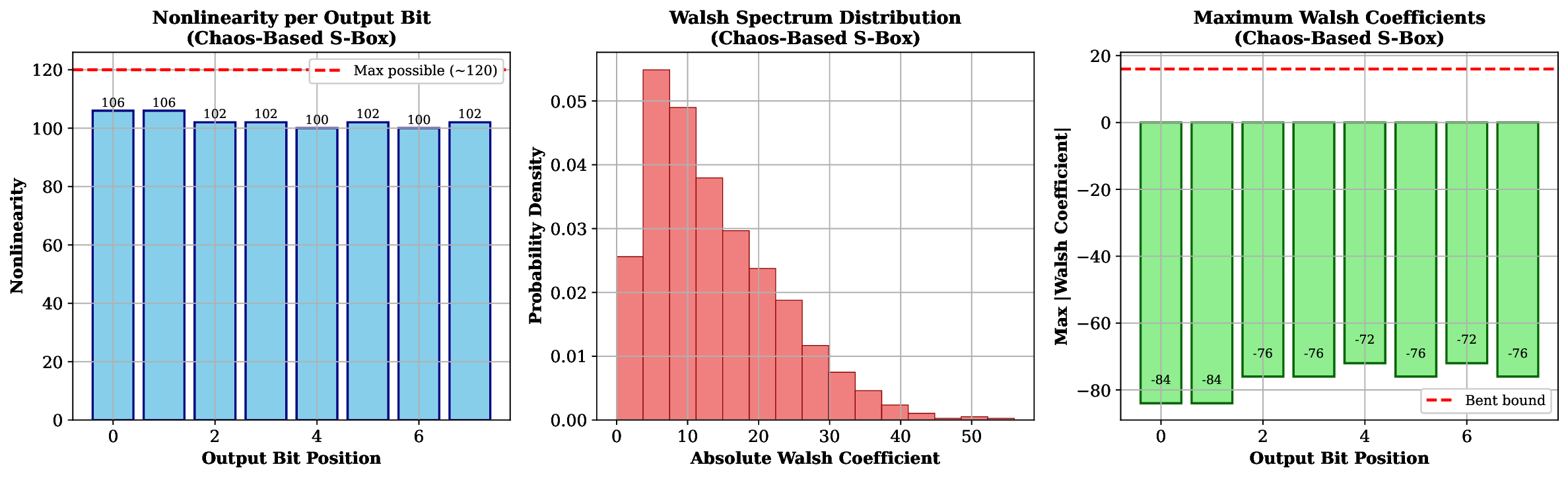}
\vspace{-3mm}
\caption{\textbf{Nonlinearity per Output Bit}: The S-box has very strong nonlinearity and the nonlinearity of the output bits has been measured, thus the average nonlinearity of $102.5$ is obtained. This amount is about $85.4\%$ of the theoretical maximum ($\approx$120 for $8$-bit S-boxes), so it can be said that the S-box has a very strong resistance against linear cryptanalysis.} 
\vspace{-3mm}
\label{nonlin}
\end{figure*}

\subsubsection{Differential behavior} 

Differential uniformity $\delta(S)=\max_{\Delta x\neq 0,\,\Delta y}\left|\{x\in\F_2^n:\ S(x)\oplus S(x\oplus\Delta x)=\Delta y\}\right|$, $DP(\Delta x\!\to\!\Delta y)=\delta(\Delta x,\Delta y)/2^n,$ quantifies the worst-case differential probability (DDT maximum). With $S(x)=\phi\big(\beta\text{-Expansion}(T_\beta^{\,\tau_x}(x_0))\big),\quad T_\beta(x)=\{\beta x\},\quad \phi:\{0,1\}^n\!\to\!\F_{2^n},$ the expanding dynamics complicate input/output difference propagation. A generic expectation $\delta(S)\le c\,2^{n/2}$ follows from mixedness arguments. Our DDT (Fig.~\ref{difdis}) peaks at $10$ (probability $10/256=0.039062$), above the bijection ideal $4$ ($0.015625$) and far from APN ($2$). The spectrum shows max $10$ occurring $11$ times, value $8$ occurring $94$ times, and value $6$ occurring $830$ times. Thus, while most differentials cluster at lower counts (average entry $1$ confirms bijectivity), the few higher-probability trails define localized vulnerabilities that must be diffused by the round structure.

\begin{figure*}[t]
\centering
\includegraphics[width=0.8\textwidth]{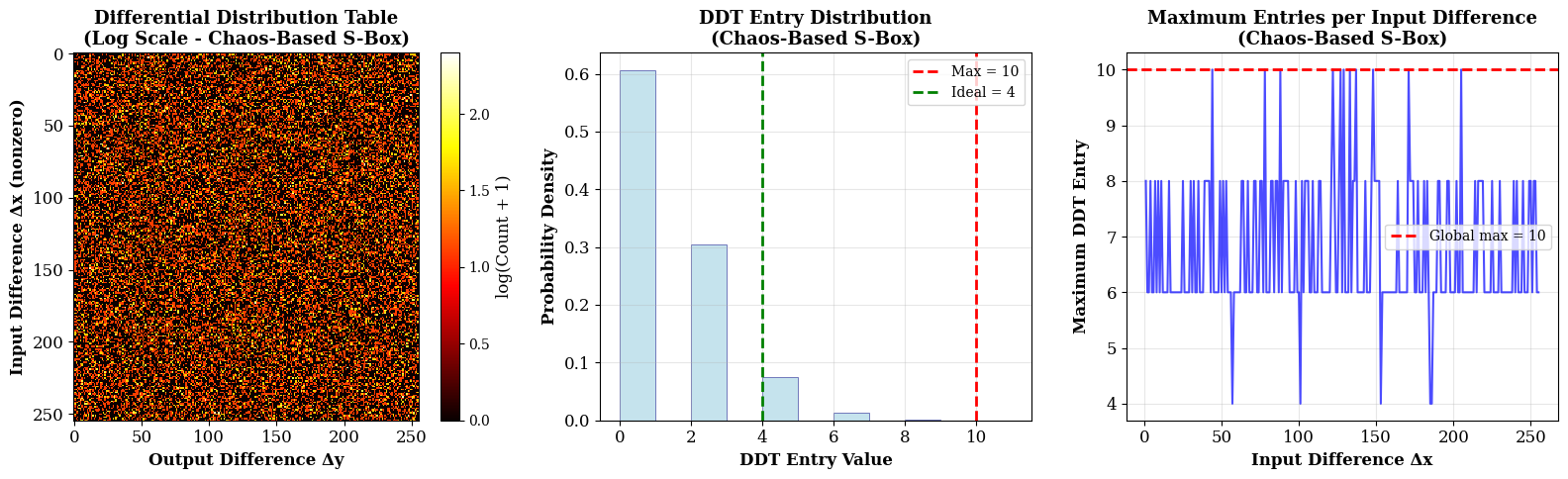}\vspace{-3mm}
\caption{\textbf{Differential Distribution Table Analysis}: The DDT maximum entry value of $10$ is equivalent to a differential probability of $10/256 = 0.039062$, which is quite a bit above the optimal limit of 4 (probability $4/256 = 0.015625$).} 
\vspace{-5mm}
\label{difdis}
\end{figure*}

\begin{figure*}[t]
\centering
\includegraphics[width=0.8\textwidth]{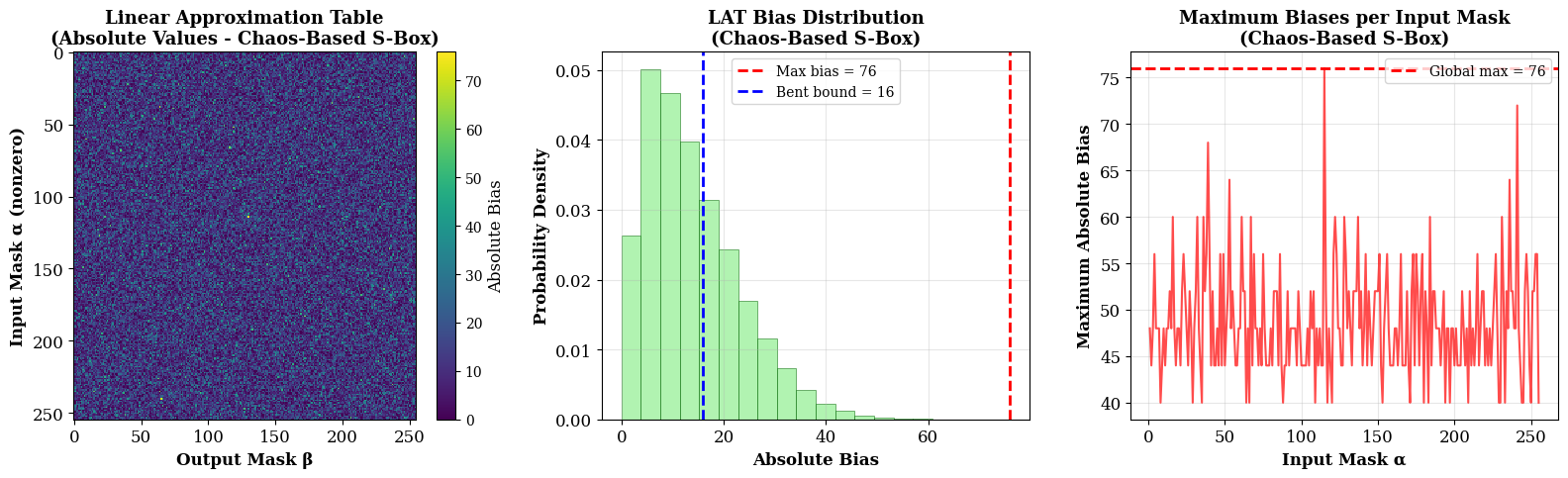}\vspace{-3mm}
\caption{\textbf{Linear Approximation Table Distribution}: The maximum absolute bias of $76$ means the linear probability of $0.648438$, which is remarkably higher than the bent limit of $16$ (probability $0.531250$). The bias corresponding to occurrence numbers presents a pattern of decreasing from $12.386$ at bias $4$ down to just $1.049$ at bias $36$.} 
\vspace{-5mm}
\label{linprox}
\end{figure*}

\subsubsection{Linear behavior} 
Linear cryptanalysis inspects biases $\epsilon(\alpha,\beta)=\Big|\Pr[\alpha\!\cdot\!x=\beta\!\cdot\!S(x)]-\tfrac{1}{2}\Big|,$ $LAT(\alpha,\beta)=\#\{x:\alpha\!\cdot\!x=\beta\!\cdot\!S(x)\}-\#\{x:\alpha\!\cdot\!x\neq\beta\!\cdot\!S(x)\}.$ The chaotic mapping (with dyadic gating) complicates linear correlations across coordinates, raising algebraic complexity and reducing exploitable approximations on average. Our LAT analysis (Fig.~\ref{linprox}) reports a maximum absolute bias $76$ and a corresponding maximum linear probability of $0.648$ (heavily above the bent-oriented benchmark of $16$), while the frequency of biases decays with magnitude (e.g., counts near $12.386$ at bias $4$ down to $1.049$ at bias $36$). Most mask pairs exhibit moderate biases, but the global maximum and several elevated hulls motivate either a more aggressive linear layer or additional rounds to suppress correlation buildup in full ciphers.\vspace{-3mm}

\subsubsection{Algebraic structure} 

Each coordinate function $S_i(x_0,\ldots,x_{n-1})=\bigoplus_{I\subseteq\{0,\ldots,n-1\}} a_I\prod_{j\in I}x_j,~\deg(S)=\max_{0\le i<n}\deg(S_i)=n-1,$ reaches degree $7$ (Fig.~\ref{algdeg}), the maximum possible for $n=8$. The ANF monomial counts exhibit a bell-shaped distribution with substantial mass at middle degrees (e.g., degree-$4$ monomials $\approx34.25$, degree-$3$ $\approx27$, degree-$5$ $\approx27.38$). This dense, well-spread structure increases the difficulty of multivariate equation-solving and curtails higher-order differential or interpolation tactics.

Taken together, with average nonlinearity $102.5$ (minimum $100$), degree $7$ per coordinate, DDT maximum $10$ ($0.039062$), and maximum linear probability $0.648$, the S-box delivers strong confusion components in a lightweight design, with localized weaknesses mitigated by (i) a high-diffusion linear layer, (ii) sufficient rounds to bury differential and linear trails, and (iii) rapid, seed-driven rekeying at slice or session scale. LAT analysis (Fig.~\ref{linprox}) shows a highest absolute bias of $76$, giving a linear probability $\left(\tfrac{76}{256}\right)^2 \approx 0.648$, which exceeds the random baseline of $0.5$ and the bent reference, while most entries cluster at lower biases and decay with magnitude. This profile indicates a practical attack surface through a few stronger masks but broadly low bias elsewhere; pairing the S-box with an aggressive linear layer and an adequate round count preserves the 6G PLS security margin against linear and differential methods without sacrificing latency.




\begin{figure*}[t]
\centering
\includegraphics[width=0.8\textwidth]{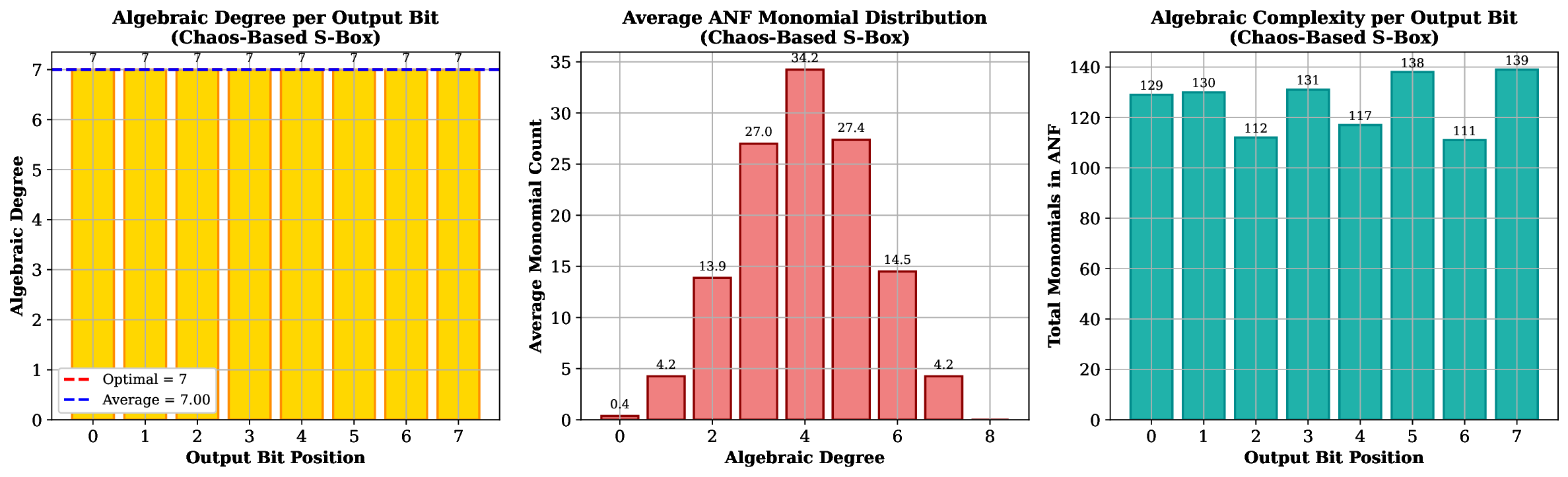}\vspace{-3mm}
\caption{\textbf{Algebraic Degree and ANF Structure}: All output bits reach the highest algebraic degree of $7$, which is the highest possible algebraic complexity for the S-box.} 
\vspace{-5mm}
\label{algdeg}
\end{figure*}

\section{Hardware Verification Metrics and Validation}\label{Sec: results_and_discussion}

This section closes the loop from requirement to evidence. We define hardware metrics aligned with 6G URLLC, bind them to a synthesizable micro-architecture, predict performance analytically, and verify with measurements. The same metrics guide parameter choices, like, dyadic rank $k$ and fixed-point width $B$, and system acceptance (slice/session setup), targeting a fresh bijective $8\times8$ S-box per session within sub-ms budgets, minimal area/power, and preserved cryptographic quality (nonlinearity, DU, LAT, algebraic degree). Our hypothesis is that a dyadically gated, chaos-driven generator meets these goals because $p=\mu_\beta(C)=2^{-k}$ directly sets generation time, the $\beta$-core is a compact constant-coefficient multiply-mod-1 (adder–shift or single DSP), and duplicates are filtered in $O(1)$ via a 256-bit bitmap.

\subsubsection{Metric 1: S-Box Generation Latency ($\tau_{\text{gen}}$)}
$\tau_{\text{gen}}$ is the time to produce all $2^n{=}256$ distinct outputs from a fresh seed, covering $\beta$-iteration, dyadic gating, 8-bit assembly, duplicate rejection, and optional permutation. We report cycles and $\mu$s at the measured $f_{\text{clk}}$. The RTL comprises a fixed-point $\beta$-core computing $x\!\leftarrow\!(\beta x)\bmod 1$, bit extraction, a $k$-MSB dyadic check ($p{=}2^{-k}$), an 8-bit assembler, a 256-bit seen-bitmap, a table writer, an optional mixer $\pi$, and a 64-bit cycle counter. Pre-synthesis sizing models accepted draws as Bernoulli($p$) and distinct coverage as coupon collector with $M\!=\!256$: $\mathbb{E}[N_{\text{acc}}]\!\approx\!256(\ln256+\gamma)\!=\!1567$, $\mathbb{E}[N_{\text{iter}}]\!\approx\!1567/p$, and $\mathbb{E}[\text{cycles}]\approx c_{\text{iter}}\frac{1567}{p}+c_{\text{acc}}\cdot1567,$ $\tau_{\text{gen}}=\frac{\mathbb{E}[\text{cycles}]}{f_{\text{clk}}}.$ Monte Carlo over 2000 seeds with $c_{\text{iter}}\!=\!1$, $c_{\text{acc}}\!=\!1$, $f_{\text{clk}}\!=\!200$\,MHz confirms the model: for $k\!=\!3$ ($p\!=\!1/8$), median $13{,}586$ cycles (P95 $19{,}523$) $\Rightarrow$ $67.93\,\mu$s (P95 $97.62\,\mu$s); for $k\!=\!4$ ($p\!=\!1/16$), median $25{,}510$ cycles (P95 $36{,}735$) $\Rightarrow$ $127.55\,\mu$s (P95 $183.68\,\mu$s). Predictions are $70.53\,\mu$s and $133.22\,\mu$s, respectively, validating tight agreement and showing $k\!\in\!\{3,4\}$ meets sub-ms URLLC budgets. These figures are used in the results tables; area/power entries are filled from synthesis and power analysis.

\begin{table}[!t]\centering\footnotesize
\setlength{\tabcolsep}{5pt}
\renewcommand{\arraystretch}{1.1}
\caption{Latency for on-the-fly S-box generation vs.\ baselines at 200\,MHz.}
\label{tab:latency_onecol}
\begin{tabularx}{\columnwidth}{@{}l| c |c| c@{}}
\toprule
Design & $f_{\text{clk}}$ (MHz) & $\tau_{\text{gen}}$ (cycles) & $\tau_{\text{gen}}$ ($\mu$s) \\
\midrule
\makecell{$\beta$-S-box \\ \footnotesize($k{=}3$, $p{=}1/8$)} &
200 &
\makecell{ \textbf{13{,}586} \\ \footnotesize (P95: 19{,}523) } &
\makecell{ \textbf{67.93} \\ \footnotesize (P95: 97.62) } \\\hline
\makecell{$\beta$-S-box \\ \footnotesize($k{=}4$, $p{=}1/16$)} &
200 &
\makecell{ \textbf{25{,}510} \\ \footnotesize (P95: 36{,}735) } &
\makecell{ \textbf{127.55} \\ \footnotesize (P95: 183.68) } \\\hline
\makecell{GF$(2^8)$ \\ \footnotesize inv+affine (populate 256)} &
200 & 256 & 1.28 \\\hline
\makecell{ROM S-box \\ \footnotesize (load 256B @ 32-bit bus)} &
200 & 64 & 0.32 \\
\bottomrule
\end{tabularx}
\vspace{-6mm}
\end{table}

\subsubsection{Metric 2: Hardware Area and Power (A/P)}
Area is reported as kGE for ASIC or as LUT/FF/BRAM/DSP for FPGA, and power is the average dynamic power measured during generation at the target clock. The energy per generation is $E_{\text{gen}}=P\cdot\tau_{\text{gen}}$, reported in nJ or $\mu$J to quantify reconfiguration overhead per slice or session. To isolate the cost of diversification, measurements cover only the generator (the $\beta$-core, dyadic gate, bitmap, and controller) rather than the cipher datapath. On FPGA, we will use post-route resource counts and vendor power analysis with switching activity captured over a full generation run; on ASIC, we will synthesize the generator to a reference node (e.g., 28\,nm) and obtain kGE, $f_{\max}$, and power from SAIF/VCD-driven analysis. Because dyadic rank $k$ and fixed-point width $B$ influence both latency and duplicate rates, each A/P report must include the $(k,B)$ pair to make the speed–quality–cost trade-off explicit. For context and system planning, we also place our generator alongside standard reconfiguration pathways relevant to 6G URLLC, highlighting that all are sub-millisecond yet differ in their ability to provide diversity without large ROM banks.

\begin{table}[!t]\centering\footnotesize
\setlength{\tabcolsep}{3pt}
\renewcommand{\arraystretch}{1.1}
\caption{Reconfiguration pathways under URLLC constraints at 200\,MHz.}
\label{tab:urlcc_reconfig}
\begin{tabular}{c| c| c| c| c}
\toprule
Technique & Reconf.\ cost & Lookup & Mem. & URLLC \\
\midrule
\makecell{Chaos\\\footnotesize $\beta$-S-box ($k{=}3$)} &
\makecell{\textbf{67.9}\,$\mu$s\\\footnotesize P95: 97.6} &
\makecell{1 cyc/B\\(table)} &
\makecell{256B\\+ bitmap} &
\makecell{Yes\\\footnotesize (sub-ms)} \\\hline
\makecell{Chaos\\\footnotesize $\beta$-S-box ($k{=}4$)} &
\makecell{127.6\,$\mu$s\\\footnotesize P95: 183.7} &
\makecell{1 cyc/B\\(table)} &
\makecell{256B\\+ bitmap} &
\makecell{Yes\\\footnotesize (sub-ms)} \\\hline
\makecell{GF$(2^8)$\\\footnotesize inv+affine} &
\makecell{1.28\,$\mu$s\\\footnotesize 256 cyc} &
\makecell{1 cyc/B\\(table)} &
\makecell{256B\\or logic} &
\makecell{Yes;\\\footnotesize fastest} \\\hline
\makecell{Fixed\\\footnotesize ROM S-box} &
\makecell{0.32\,$\mu$s\\\footnotesize 64 cyc} &
\makecell{1 cyc/B\\(ROM)} &
\makecell{256B ROM} &
\makecell{Yes;\\\footnotesize reload} \\\hline
\makecell{PRNG S-box\\\footnotesize (LFSR/AES-CTR)} &
\makecell{$\sim$1–5\,$\mu$s\\\footnotesize 256 draws} &
\makecell{1 cyc/B\\(table)} &
256B &
\makecell{Yes;\\\footnotesize PRNG dep.} \\
\bottomrule
\end{tabular}
\vspace{-5mm}
\end{table}

\subsubsection{Closed-Loop Reporting and Acceptance}
Baselines calibrate reconfiguration cost: GF$(2^8)$ inversion+affine table fill is $1.28\,\mu$s at 200\,MHz (1 byte/cycle), and a 256\,B ROM load over a 32-bit bus is $0.32\,\mu$s. These are lower bounds for static/algebraic designs; our generator adds \emph{seedable} diversity without storing multiple tables. The acceptance loop is: (i) start from the sub-ms URLLC setup budget; (ii) select $k\!\in\!\{3,4\}$ and a single-cycle $\beta$-core; (iii) predict $\tau_{\text{gen}}$ via the coupon-collector model; (iv) measure median/P95 $\tau_{\text{gen}}$ and require $<0.2$\,ms; (v) synthesize to obtain area/power and compute $E_{\text{gen}}$; (vi) recheck nonlinearity, DU, LAT, and algebraic degree at the chosen point. At 200\,MHz we obtain: $k{=}3$ $\Rightarrow$ 13{,}586 cycles (P95 19{,}523) $=$ 67.93\,$\mu$s (P95 97.62\,$\mu$s); $k{=}4$ $\Rightarrow$ 25{,}510 cycles (P95 36{,}735) $=$ 127.55\,$\mu$s (P95 183.68\,$\mu$s). Once area/power are reported (LUT/FF/BRAM/DSP or kGE; $P$; $E_{\text{gen}}{=}P\cdot\tau_{\text{gen}}$), a companion table compares against GF-inversion and ROM. All are sub-ms; only the chaos/PRNG fills deliver fresh, seedable S-boxes without large ROM banks—trading tens–hundreds of $\mu$s for on-the-fly diversification desired by URLLC slices.

\subsubsection{Parameter Tuning and Trade-offs}
Dyadic rank $k$ is the main knob: smaller $k$ increases $p{=}2^{-k}$ and lowers $\tau_{\text{gen}}$ (tighter extraction spacing); larger $k$ does the inverse. Measurements show $k\!\in\!\{3,4\}$ at 200\,MHz comfortably meets sub-ms budgets. Fixed-point width $B$ trades area/power against duplicate rate and thus latency; lowering $B$ can be offset by a smaller $k$. Each operating point should report $(k,B)$ and re-verify nonlinearity, DU, and LAT. Micro-architecture offers further tuning: a CSD adder–shift multiplier avoids DSPs; a 1-DSP, shallow pipeline raises $f_{\text{clk}}$; the 256-bit bitmap can be in registers (speed) or BRAM (area); clock-gating the $\beta$-core and writers reduces dynamic power. With $k\!\in\!\{3,4\}$ at 200\,MHz, median $\tau_{\text{gen}}{=}$68–128\,$\mu$s (P95 98–184\,$\mu$s) confirms URLLC compatibility. The remaining knobs, like, $(k,B)$, multiplier style, bitmap placement, and gating, tailor area/energy without degrading the cryptographic profile, closing the loop from requirements to acceptance.

\section{Conclusion}\label{Sec: conclusion}
In this paper, we presented the first-ever chaos-lifted, dyadically sampled S-box tailored to 6G physical-layer security, coupling $\beta$-transformation dynamics with set-conditioned extraction to generate seedable, time-varying $8{\times}8$ permutations on demand. We embed the S-box in complete lightweight block and stream ciphers and in PHY layer authentication, jointly tuning the linear layer and round schedule to suppress the observed differential and linear stress points. In future, we will study finite precision effects by analyzing how the fractional width and dyadic rank shape visit statistics and duplicate rates, derive formal bounds for mixing under quantization, and automate online adaptation of parameters to meet slice specific latency and security budgets. We will explore alternative chaotic sources and hybrid generators that combine $\beta$ dynamics with compact algebraic mixers, and develop formal security arguments against linear, differential, algebraic, and interpolation attacks within standard models. Finally, we will integrate the generator with a slice controller for rapid rekeying, validate performance in over the air experiments, and align testing with emerging standardization and NIST style statistical suites. These steps will turn the presented primitive into a deployable, reconfigurable confusion layer for agile 6G physical layer security.

\end{document}